\documentclass[conference]{IEEEtran}
\IEEEoverridecommandlockouts
\usepackage{cite}
\usepackage{amsmath,amssymb,amsfonts}
\usepackage{algorithmic}
\usepackage{graphicx}
\usepackage{textcomp}
\usepackage{xcolor}
\def\BibTeX{{\rm B\kern-.05em{\sc i\kern-.025em b}\kern-.08em
    T\kern-.1667em\lower.7ex\hbox{E}\kern-.125emX}}

\usepackage{graphicx}
\usepackage{epstopdf}
\usepackage{amsmath,amssymb,amsthm,mathrsfs,amsfonts,dsfont}
\usepackage{epsfig}
\usepackage[ruled,linesnumbered]{algorithm2e}
\usepackage{color}
\usepackage{subfigure}
\usepackage{cite}
\usepackage{diagbox}
\usepackage{multirow}
\usepackage{float}
\usepackage{stfloats}
\usepackage{enumitem}
\usepackage{cases}
\usepackage{setspace}
\usepackage{adjustbox,lipsum}
\usepackage{comment}
\usepackage{lipsum}

\newcommand{\ba}{\begin{array}}
\newcommand{\ea}{\end{array}}
\newcommand{\be}{\begin{displaymath}}
\newcommand{\ee}{\end{displaymath}}
\newcommand{\ben}{\begin{equation}}
\newcommand{\een}{\end{equation}}
\newcommand{\bena}{\begin{eqnarray}}
\newcommand{\eena}{\end{eqnarray}}
\newcommand{\beqa}{\begin{eqnarray*}}
\newcommand{\enqa}{\end{eqnarray*}}

\newcommand{\bc}{\begin{center}}
\newcommand{\ec}{\end{center}}
\newcommand{\bi}{\begin{itemize}}
\newcommand{\ei}{\end{itemize}}
\newcommand{\benu}{\begin{enumerate}}
\newcommand{\eenu}{\end{enumerate}}
\newcommand{\bdes}{\begin{description}}
\newcommand{\edes}{\end{description}}
\newcommand{\bt}{\begin{tabular}}
\newcommand{\et}{\end{tabular}}

\newcommand \mubf{\mbox{\boldmath$\mu$\unboldmath}}

\newcommand \hbf{{\bf h}}

\newcommand \nbf{{\bf n}}

\newcommand \vbf{{\bf v}}
\newcommand \wbf{{\bf w}}
\newcommand \xbf{{\bf x}}
\newcommand \ybf{{\bf y}}

\newcommand \Cbf{{\bf C}}

\newcommand \Fbf{{\bf F}}

\newcommand \Hbf{{\bf H}}
\newcommand \Ibf{{\bf I}}

\newcommand{\circlambda}{\mbox{$\Lambda$
             \kern-.85em\raise1.5ex
             \hbox{$\scriptstyle{\circ}$}}\,}

\graphicspath{{figures/}}

\newcommand{\mypara}[1]{{\smallskip \noindent \bf #1}\hspace{0.1in}}
\newcommand{\ssf}[1]{\textrm{$\sf{#1}$}{}}
\newcommand{\real}{\text{Re}}

\newtheorem{theorem}{Theorem}

\newtheorem{lemma}{Lemma}
\newtheorem{assumption}{Assumption}

\newtheorem{remark}{Remark}

\DeclareMathOperator*{\argmin}{arg\,min}

\newcommand{\vect}[1]{\mathbf{#1}}

\newcommand{\expt}{\mathbb{E}}
\newcommand{\variance}{\mathbb{V}\text{ar}}
\newcommand{\norm}[1]{\left \| #1 \right \|}
\newcommand{\squab}[1]{\left [ #1 \right ]}

\ifodd 1

\else

\fi

\ifodd 1
\newcommand{\congc}[1]{{\color{red}(Cong: #1)}}
\else
\newcommand{\congc}[1]{}
\fi

\ifodd 0
\newcommand{\zixiang}[1]{{\color{blue}(Zixiang: #1)}}
\else
\newcommand{\zixiang}[1]{}
\fi

\ifodd 1

\else

\fi

\begin{document}


\title{Random Orthogonalization for Federated Learning in Massive MIMO Systems}

\author{\IEEEauthorblockN{Xizixiang Wei$^*$, Cong Shen$^*$, Jing Yang$^\dag$, H. Vincent Poor$^\ddag$}
\IEEEauthorblockA{$^*$ Department of Electrical and Computer Engineering, University of Virginia, USA\\
$^\dag$ Department of Electrical Engineering, The Pennsylvania State University, USA\\
$^\ddag$ Department of Electrical and Computer Engineering, Princeton University, USA
}
}

\maketitle

\begin{abstract}
We propose a novel uplink communication method, coined \emph{random orthogonalization}, for federated learning (FL) in a massive multiple-input and multiple-output (MIMO) wireless system. The key novelty of random orthogonalization comes from the tight coupling of FL model aggregation and two unique characteristics of massive MIMO -- channel hardening and favorable propagation. As a result, random orthogonalization can achieve natural over-the-air model aggregation without requiring transmitter side channel state information, while significantly reducing the channel estimation overhead at the receiver. Theoretical analyses with respect to both communication and machine learning performances are carried out. In particular, an explicit relationship among the convergence rate, the number of clients and the number of antennas is established. Experimental results validate the effectiveness and efficiency of random orthogonalization for FL in massive MIMO.

\end{abstract}


\section{Introduction}


Communication overhead is widely considered one of the primary bottlenecks for federated learning (FL) \cite{mcmahan2017fl,konecny2016fl}, as a FL task consists of multiple learning rounds, each of which requires uplink and downlink model exchange between clients and the server. Compared with downlink broadcasting, uplink communication is more challenging in FL. Due to the strigent power constraint at edge devices, channel noise and fading have more conspicuous impacts on uplink communications. More importantly, the limited uplink communication resources may severely limit the \emph{scalability} of FL, negatively affecting one of its primary features \cite{kairouz2019advances}.

To tackle the scalability problem in FL uplink communications, several over-the-air computation (also known as \emph{AirComp}) mechanisms have been exploited in wireless FL (see \cite{niknam2020federated} and the references therein). Instead of decoding the individual local models of each client and then aggregating, AirComp allows multiple clients to transmit uplink signals in a superpositioned fashion, and decodes the average model (global model) directly at the FL server. Zhu et al. \cite{zhu2019broadband} propose an analog aggregation framework which ``inverts'' the fading channel at each transmitter, so that the sum model can be directly obtained at the server. However, the fundamental limitation of analog aggregation is that it requires channel state information at transmitter (CSIT). The process of enabling CSIT is complicated and the precision of CSIT is often worse than the channel state information at receiver (CSIR). Besides, analog aggregation essentially requires a channel inversion power control, which is well known to ``blow up'' when channel is in deep fade. Moreover, analog aggregation does not naturally extend to multiple-input and multiple-output (MIMO) systems where the uplink channels become vectors, which makes channel inversions at the transmitters nontrivial.  

This paper aims at designing a simple-yet-effective uplink FL communication and model aggregation method. To address the scalability challenge in FL, we explore another design degree of freedom (d.o.f.) in modern wireless systems: \emph{massive MIMO}. The proposed framework only requires the BS to estimate a summation channel, which significantly alleviates the burden on uplink channel estimation in FL. Moreover, this approach is agnostic to the number of clients, making it attractive for the scalability of FL.  By tightly integrating the channel hardening and favorable propagation properties of massive MIMO, the proposed scheme, coined \emph{random orthogonalization}, allows the BS to directly compute the global model via a simple linear projection operation, thus achieving extremely low complexity and low latency. To analyze the performances of random orthogonalization, we derive the Cramer-Rao lower bounds (CRLBs) of the average model estimation as a theoretical benchmark. Moreover, taking both interference and noise into consideration, a novel convergence bound of FL is derived for the proposed method over massive MIMO channels. Notably, we establish an explicit relationship among the convergence rate, the number of clients $K$, and the number of antennas $M$, which provides practical design guidance for wireless FL. Numerical results validate the effectiveness and efficiency of the proposed method.

The potential of MIMO for wireless FL has attracted interest recently. 
MIMO beamforming design to optimize FL has been studied in \cite{yang2020federated,elbir2020federated}. Coding, quantization, and compressive sensing over a (massive) MIMO channel for FL has been studied in \cite{huang2020physical,jeon2020compressive,jeon2020gradient}. Nevertheless, none of these works tightly incorporates the unique properties of massive MIMO in the FL uplink communication design. On the other hand, massive MIMO can also be utilized in a straightforward manner, e.g., one can use traditional MIMO decoders such as zero-forcing (ZF) or minimum mean-square-error (MMSE) to estimate each local model, and then compute the global model. However, this heuristic approach requires large channel estimation overhead, especially in massive MIMO. Decoding individual local models also makes it easier for the server to sketch the data distribution of a client. Moreover, matrix inversion operations in ZF or MMSE detectors are computationally demanding, which increases the complexity and latency.


The remainder of this paper is organized as follows. Section \ref{sec:model} introduces the FL pipeline and the uplink communication model. The proposed random orthogonalization design is detailed in Section \ref{sec:RO}. The CRLB evaluation along with the model convergence analysis are presented in Section \ref{sec:analysis}. Experimental results are reported in Section \ref{sec:sim}, followed by the conclusion of our work in Section \ref{sec:conc}.

\section{System Model}
\label{sec:model}

\subsection{FL Model}
Consider a FL task with a central server and $K$ clients. Each client $k \in [K]$ stores a (disjoint) local dataset $\mathcal{D}_k$, with its size denoted by $D_k$. The size of the total data is $D \triangleq \sum_{k\in [K]} D_k$. We use $f_k(\wbf)$ to denote the local loss function at client $k$, which measures how well a machine learning (ML) model with parameter $\wbf \in \mathbb{R}^d$ fits its local dataset. The global objective function over all $K$ clients is
$
    f(\wbf) = \sum_{k \in [K]} p_k f_k(\wbf),
$
where $p_k = \frac{D_k}{D}$ is the weight of each local loss function, and the purpose of FL is to distributively find the optimal model parameter $\wbf^*$ that minimizes the global loss function: 
$
    \wbf^* \triangleq \argmin_{\wbf\in\mathbb{R}^d}f(\wbf).
$
A typical wireless FL pipeline is illustrated in Fig.~\ref{fig:FLpipeline}. Specifically, this pipeline iteratively executes the following steps at the $t$-th learning round.
\begin{enumerate}
\item \textbf{Downlink communication.} The BS broadcasts the current global model $\wbf_t$ to all devices over the downlink wireless channel.
\item \textbf{Local computation.} Each client uses its local data to train a local model improved upon the received global model $\wbf_t$. We assume that mini-batch stochastic gradient descent (SGD) is used to minimize the local loss function. The parameter is updated iteratively (for $E$ steps) at client $k$ as: $\wbf_{t,0}^k = \wbf_t; \wbf_{t,\tau}^k = \wbf_{t,\tau-1}^k - \eta_t \nabla \tilde{f}_k(\wbf_{t,\tau - 1}^k); \forall \tau = 1, \cdots, E; \wbf_{t+1}^k = \wbf_{t,E}^k$, where $\nabla\tilde{f}_k(\wbf)$ denotes the mini-batch SGD operation at client $k$ on model $\wbf$.
\item \textbf{Uplink communication.} Each client uploads its latest local model to the server synchronously over the uplink wireless channel.
\item \textbf{Server Aggregation.} The BS aggregates the received noisy local models $\tilde \wbf_{t+1}^k$ to generate a new global model: $\wbf_{t+1} = \Sigma_{k\in [K]} p_{k} \tilde \wbf_{t+1}^k$. For simplicity, we assume that each local dataset has equal size, hence $p_{k} = \frac{1}{K}$. 
\end{enumerate}

This work focuses on steps $3$ and $4$ in the FL pipeline. In particular, we take advantage of the unique properties of massive MIMO to design efficient FL uplink communication and server aggregation. 

\begin{figure}[t]
    \centering
    \includegraphics[width = 0.8\linewidth]{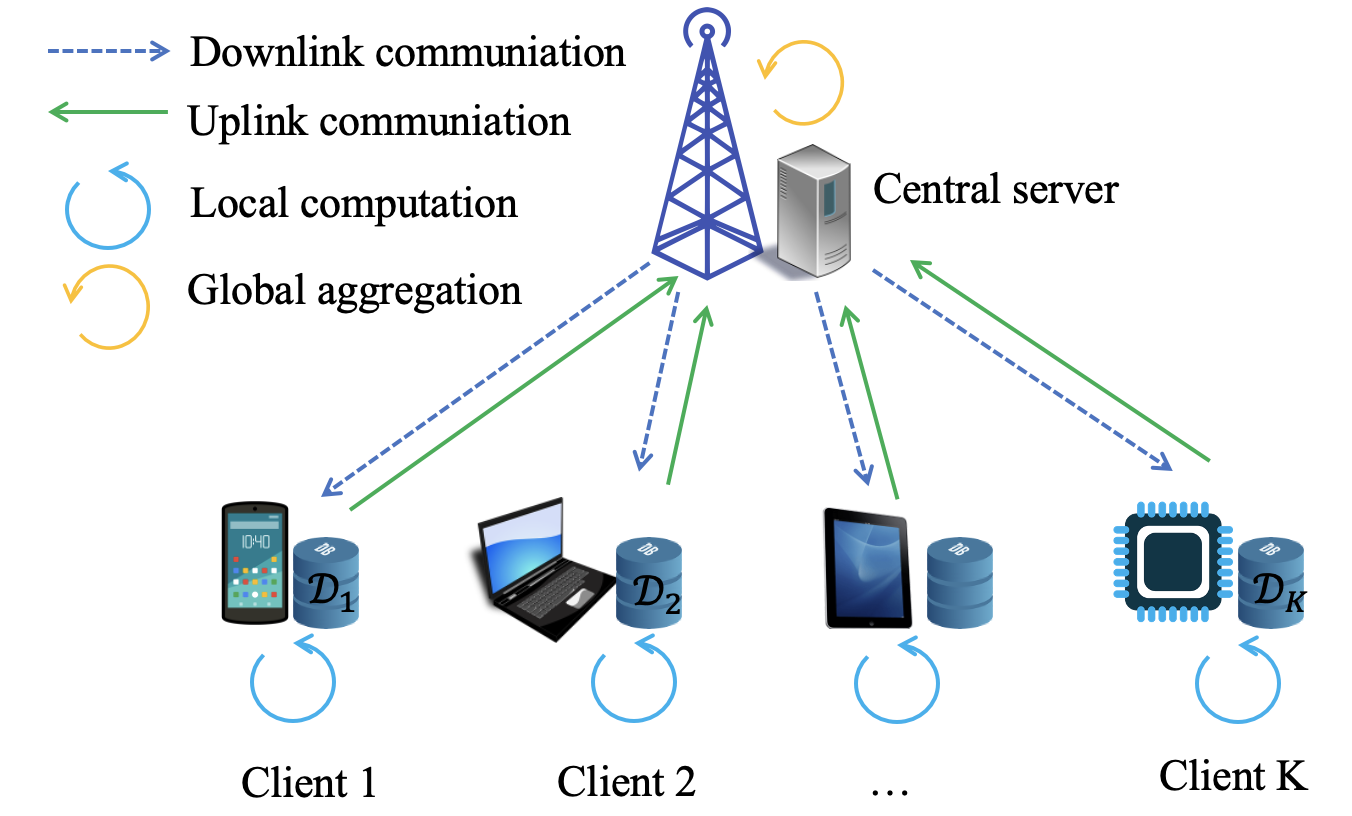}
    \caption{The wireless FL pipeline.}
    \vspace{-0.2in}
    \label{fig:FLpipeline}
\end{figure}

\subsection{Communication Model}
Consider a massive MIMO system equipped with $M$ antennas at the BS (server) where $K$ single-antenna devices (clients) are involved in the aforementioned FL task. At the uplink step of the $t$-th round, each client transmits the differential between the received global model and the computed new local model 
${\mathsf{x}_t^k} = \wbf_t - \wbf_{t+1}^k\in\mathbb{R}^d, \forall k \in [K]$  
to the BS\footnote{The parameter normalization and de-normalization procedure in wireless FL follows the same as that in the Appendix of \cite{zhu2019broadband}.}, where ${\mathsf{x}_t^k} \triangleq [x_{1,t}^k, \cdots, x_{i,t}^k,\cdots, x_{d,t}^k]^T$. To simplify the notation, we omit index $t$ by using  $x_{k,i}$ instead of $x_{i,t}^k$ barring any confusion. We assume that each client transmits every element of the differential model $\{x_{k,i}\}_{i = 1}^d$ via $d$ shared time slots\footnote{In general, differential model parameters can be transmitted over any $d$ shared time-frequency resources. For simplicity, we use $d$ time slots here.}. For a given element $x_{k,i}$, the received signal at the BS is
\begin{equation}\label{eq:FLuplink}
    \ybf_i = \sqrt{P}\sum_{k\in[K]} {\bf h}_k x_{k,i} +  \nbf_i,\;\;\forall i = 1,\cdots,d, 
\end{equation}
where $P$ is the maximum transmit power of each client, $\hbf_k\in\mathbb{C}^{M\times1}$ is the wireless channel between $k$-th client and BS, and $\nbf_i\in\mathbb{C}^{M\times1}$ is the uplink noise. We assume normalized symbol power $\expt\norm{x_{k,i}}^2 = 1$, {normalized} Rayleigh block fading channel\footnote{{Large-scale pathloss and shadowing effect is assumed to be taken care of by, e.g., open loop power control \cite{SesiaLTE}.} } $\hbf_k\sim\mathcal{CN}(0,{\frac{1}{M}\bf I})$ in $d$ slots, and independent and identically distributed (i.i.d.) Gaussian noise $\nbf_i \sim\mathcal{CN}(0,\sigma^2{\bf I})$. 
We define the signal-to-noise ratio (SNR) as $\ssf{SNR} \triangleq {P}/{\sigma^2}$, and w.l.o.g. we set $P = 1$. 
Denoting $\Hbf \triangleq \squab{\hbf_1, \cdots, \hbf_K}\in\mathbb{C}^{M\times K}$ and $\xbf_i \triangleq \squab{x_{1,i},\cdots,x_{K,i}}^T \in\mathbb{R}^{K\times 1}, \forall i = 1,\cdots, d$, the received signal\footnote{For simplicity, we assume real signals $\{x_{k,i}\}_{i = 1}^d$ are transmitted in this paper. It can be easily extended to complex signals by stacking two real model parameters into a complex signal, so that the full d.o.f. is utilized.} can be written as
\begin{equation}\label{eq:MatrixForm}
    \ybf_i = \Hbf \xbf_i + \nbf_i.
\end{equation}
Eqn.~\eqref{eq:MatrixForm} is a standard MIMO model and traditional MIMO decoders can be adopted to estimate $\hat\xbf_i = \squab{\hat x_{1,i},\cdots,\hat x_{K,i}}^T$. However, as discussed before, decoding $\{x_{k,i}\}_{i = 1}^d$ individually and obtaining the aggregated parameter $\tilde{x}_i\triangleq\sum_{k\in[K]} \hat x_{k,i}$ by a summation is inefficient. We propose a novel method that allows the BS to compute $\tilde{x}_i$ directly. Note that after BS decoding all aggregated parameter $\tilde \xbf_t \triangleq \squab{\tilde{x}_1, \cdots,\tilde{x}_d}^T$ in $d$ slots, it can compute the new global model as
\begin{equation}\label{eq:diffGlobal}
    \wbf_{t + 1} = \wbf_t + \frac{1}{K}\tilde \xbf_t.
\end{equation}

\begin{figure*}
    \centering
    \includegraphics[width = 0.9 \linewidth]{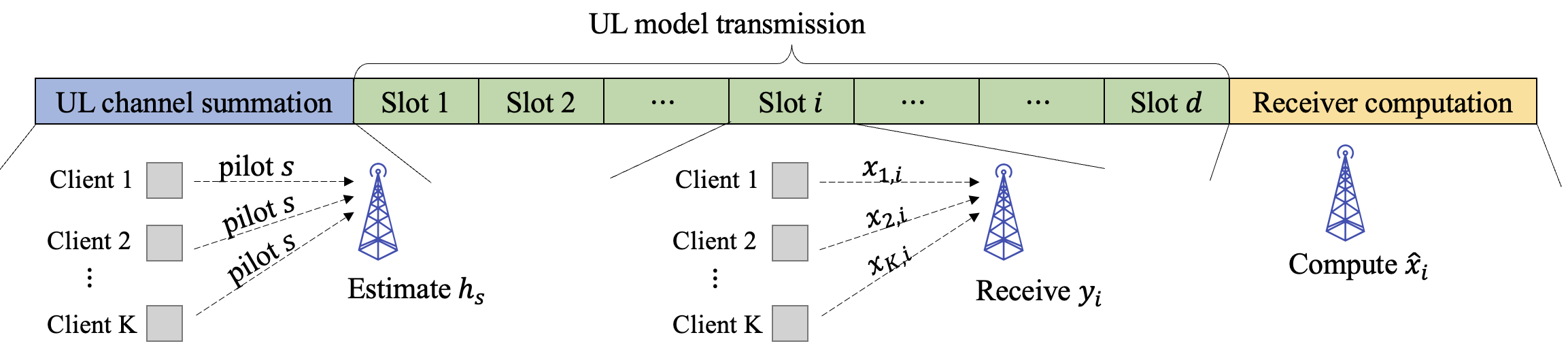}
    \caption{An illustration of the proposed uplink FL design with massive MIMO.}
    \label{fig:framework}
    \vspace{-0.2in}
\end{figure*}

\section{Random Orthogonalization}
\label{sec:RO}
We study a wireless FL framework where the global model can be directly obtained at the BS via a simple operation. By exploring favorable propagation and channel hardening in massive MIMO, our proposed FL framework obtains the global model by the following three main steps.

\mypara{(1) Uplink channel summation.}
The BS first schedules all participating clients to transmit a \emph{common} pilot signal $s$ synchronously. The received signal at the BS is
\begin{equation*}
    \ybf_s = \sum_{k\in[K]} \hbf_k s + \nbf_s,
\end{equation*}
so that the BS can estimate the {\em summation} of channel {vectors} $\hbf_s \triangleq \sum_{k \in [K]} \hbf_k$ from the received signal $\ybf_s$ (e.g., via a maximum likelihood estimator). For simplicity, we assume perfect sum channel estimation at the BS.

\mypara{(2) Uplink model transmission.} All clients transmit model differential parameters $\{x_{k,i}\}_{i=1}^d$ to the BS in $d$ time slots. The received signal for each differential model element is 
\begin{equation*}
    \ybf_{i} = \sum_{k \in [K]} \hbf_k x_{k,i} + \nbf_i,\;\;\forall i = 1,\cdots, d.
\end{equation*}

\mypara{(3) Receiver computation.} The BS estimates each aggregated parameter via a simple \emph{linear projection} operation:
\begin{align}
    & \tilde x_i = \hbf_s^H \ybf_{i} = \sum_{k \in [K]} \hbf_k^H\sum_{k \in [K]} \hbf_k x_{k,i} +  \sum_{k \in [K]} \hbf_k^H\nbf_i \nonumber\\
    & \overset{(a)}{=} \underbrace{\sum_{k \in [K]} \hbf_k^H \hbf_k x_{k,i}}_{\text{Signal}} + \underbrace{\sum_{k \in [K]}\sum_{j \in [K], j\neq k} \hbf_k^H \hbf_j x_{j,i}}_{\text{Interference}}  +  \underbrace{\sum_{k \in [K]} \hbf_k^H\nbf_i}_{\text{noise}} \nonumber \\
    & \overset{(b)}{\approx} \sum_{k \in [K]} x_{k,i},\;\;\forall i = 1,\cdots, d. \label{eq:OTAComp}
\end{align}
The above three-step procedure is illustrated in Fig.~\ref{fig:framework}. Based on Eqn.~\eqref{eq:OTAComp}, BS then computes the global model via Eqn.~\eqref{eq:diffGlobal} and begins the next communication round.
As shown in part (a) of Eqn.~\eqref{eq:OTAComp}, inner product $\hbf_s^H \ybf_{i}$ can be regarded as the combination of three parts: signal, interference, and noise. We next show that, taking advantage of two fundamental properties of massive MIMO, the approximation (b) in Eqn.~\eqref{eq:OTAComp} is asymptotically error-free, as the number of antennas at the BS $M$ goes to infinity. 

\mypara{Channel hardening.} Since each element of $\hbf_k$ is i.i.d. complex Gaussian, by the law of large numbers, massive MIMO enjoys channel hardening \cite{ngo2014aspects}:
\begin{equation*}\label{eq:ChannelHardening}\small
    \hbf_k^H\hbf_k\rightarrow 1,\;\;\text{as}\;\;M\rightarrow\infty.
\end{equation*}
In practical systems, when $M$ is large but finite, we have
\begin{equation}
    \expt_{\hbf}\squab{\sum_{k \in [K]} \hbf_k^H \hbf_k x_{k,i}} = \sum_{k \in [K]}x_{k,i},
\end{equation}
and
\begin{equation}\label{eq.varhkk}\small
    \variance_{\hbf}\squab{\sum_{k \in [K]} \hbf_k^H \hbf_k x_{k,i}} = \frac{\sum_{k \in [K]}x_{k,i}^2}{M}
\end{equation}
for the signal part of (\ref{eq:OTAComp}).

\mypara{Favorable propagation.} Since channels between different users are independent random vectors, massive MIMO also offers favorable propagation \cite{ngo2014aspects}:
\begin{equation*}
    \hbf_k^H\hbf_j\rightarrow 0,\;\; \text{as} \;\; M\rightarrow\infty,\;\;\forall k\neq j.
\end{equation*}
Similarly, when $M$ is finite, we have
\begin{equation}\small
    \expt_{\hbf}\squab{\sum_{k \in [K]}\sum_{j \in [K], j\neq k} \hbf_k^H \hbf_j x_{j,i}} = 0,
\end{equation}
and
\begin{equation}\label{eq.varhkj}\small
    \variance_{\hbf}\squab{\sum_{k \in [K]}\sum_{j \in [K], j\neq k} \hbf_k^H \hbf_j x_{j,i}} = \frac{(K-1)\sum_{k \in [K]}x^2_{k,i}}{M}.
\end{equation}

Furthermore, the expectation of the noise part in \eqref{eq:OTAComp} is zero. Therefore, $\tilde{x}_i$ in Eqn.~\eqref{eq:OTAComp} is an unbiased estimate of the average model. For a given $K$, the variances of both signal and interference decrease in the order of $\mathcal{O}(1/M)$, which shows that \textit{massive MIMO offers {\bf random orthogonality} for analog aggregation over wireless channels}. In particular, the asymptotic element-wise orthogonality of channel vector ensures channel hardening, and the asymptotic vector-wise orthogonality among different wireless channel vectors provides favorable propagation, which make the linear projection operation $\hbf_s^H \ybf_{i}$ an ideal fit for FL. 

To gain some insight of random orthogonality, we approximate the average signal-to-interference-plus-noise-ratio (SINR) after the operation in Eqn.~\eqref{eq:OTAComp} as
\begin{equation}\label{eq:SINR} \small
\begin{split}
    &\expt[\text{SINR}_i] \approx \\&  \frac{\expt_{\hbf,x}\norm{\sum_{k \in [K]} \hbf_k^H \hbf_k x_{k,i}}^2}{\expt_{\hbf,\nbf,x}\norm{\sum_{k \in [K]}\sum_{j \in [K], j\neq k} \hbf_k^H \hbf_j x_{j,i} + \sum_{k \in [K]} \hbf_k^H\nbf_i}^2}\\
    & = \frac{M}{K - 1 + 1/\ssf{SNR}},
\end{split}
\end{equation}
which grows linearly with $M$ for a fixed $K$. 
We note that Eqn.~\eqref{eq:SINR} is an approximate expression for SINR but it sheds light into the relationship between $K$ and $M$. This approximation, however, is not used in the convergence analysis of FL with random orthogonalization in Section \ref{subsec:CvgAna}.

\begin{remark}
Unlike the analog aggregation method in \cite{zhu2019broadband}, random orthogonalization does not require any CSIT, and only requires the receiver to estimate a summation channel $\hbf_s$, which is $1/K$ of the channel estimation overhead compared with the AirComp method in \cite{yang2020federated} and traditional MIMO decoders. Moreover, the global model is obtained after a single linear projection, which improves the privacy and reduces the system latency.
\end{remark}

\begin{remark}
The proposed framework assumes a perfect estimation of $\hbf_s$ and requires channel hardening and favorable propagation. In practical systems, to improve the accuracy of the estimate $\hat{\hbf}_s$, BS can adopt multiple pilots for channel estimation. We will provide more details on the robustness of the proposed scheme over imperfect $\hat{\hbf}_s$ and evaluate the circumstances where channel hardening and favorable propagation are not fully offered, e.g. correlated channels, in the journal version.
\end{remark}

\section{Performance Analysis}
\label{sec:analysis}

In this section, we analyze the performances of random orthogonalization in FL. We first derive CRLBs of the estimates of global model parameters as the theoretical benchmark of the proposed method. Then, by {an} ML model convergence analysis, we investigate the relationship between the number of involved clients $K$ and the number of BS antennas $M$. We show that random orthogonalization has the potential to achieve nearly the same convergence rate as the interference-free case in massive MIMO systems. 

\subsection{Cramer-Rao Lower Bounds}
Recall that the received signal is $\ybf_{i} = \Hbf\xbf_i + \nbf_i$. Denoting $\mubf = \Hbf\xbf_i$, we have that $\ybf_{i} \sim \mathcal{CN}(\mubf, \frac{1}{\ssf{SNR}}\Ibf)$. The Fisher information matrix (FIM) of the estimation of $ \xbf_i$ is 
\begin{equation*}
    \Fbf = 2\cdot  \ssf{SNR} \cdot  \real\squab{\frac{\partial^H\mubf(\xbf_i)}{\partial \xbf_i}\frac{\partial\mubf(\xbf_i)}{\partial \xbf_i}}.
\end{equation*}
After inserting $\frac{\partial\mubf(\xbf_i)}{\partial \xbf_i} = \Hbf$ into FIM, we have $\Fbf = 2 \cdot \ssf{SNR} \cdot \real(\Hbf^H\Hbf)$. The CRLBs are given by the inverse of the FIM
    $\Cbf_{\hat\xbf_i} = \Fbf^{-1}.$ 
CRLB expresses a lower bound on the variance of unbiased estimators, stating that the variance of any such estimator is at least as high as the inverse of the FIM. Eqn.~\eqref{eq:OTAComp} has shown that the proposed method leads to an unbiased estimation of the global model; hence we can use the sum of all diagonal elements of $\Cbf_{\hat\xbf}$ as the lower bound of the mean squared error (MSE) $\expt\norm{\xbf_i - \hat\xbf_i}^2$ to evaluate the performance of global model estimation. 

\subsection{Convergence analysis}\label{subsec:CvgAna}
To simplify the analysis, we assume\footnote{We will address the general case of $E>1$ in the journal version.} $E = 1$, which is also referred to as \emph{parallel SGD} \cite{stich2018local}, and make the following standard assumptions that are commonly adopted in the convergence analysis of \textsc{FedAvg} and its variants; see \cite{li2019convergence,jiang2018nips,stich2018local,Zheng2020jsac}.  
\begin{assumption}\label{as:1}
\textbf{$L$-smooth:} $\forall~\vect{v}$ and $\vect{w}$, $\norm{f_k(\vbf)-f_k(\wbf)}\leq L \norm{\vbf-\wbf}$;
\end{assumption}
\begin{assumption}\label{as:2}
\textbf{$\mu$-strongly convex:} $\forall~\vect{v}$ and $\vect{w}$, $\left<f_k(\vbf)-f_k(\wbf), \vbf-\wbf\right>\geq \mu \norm{\vbf-\wbf}^2$;
\begin{assumption}\label{as:3}
\textbf{Unbiased SGD:} $\forall k \in [K]$, $\expt[\nabla\tilde{f}_k(\wbf)] = \nabla f_k(\wbf)$;
\end{assumption}
\begin{assumption}\label{as:4}
\textbf{Uniformly bounded gradient:} $\forall k \in [K]$, $\expt\norm{\nabla\tilde{f}_k(\wbf)}^2 \leq H^2$ for all mini-batch data.
\end{assumption}
\end{assumption}
\begin{lemma}[\textbf{\em One-step convergence}]\label{lemma.1}
Based on Assumptions 1-4 and selecting learning rate $\eta_t\leq 1/(2\mu)$, $\forall t\in [T]$, the following inequality holds for parallel SGD:
\begin{equation}\label{eq.thrm1}\small
\begin{split}
    \expt\norm{\wbf_{t+1} - \wbf^*}^2& \leq(1-2\mu\eta_t)\expt\norm{\wbf_t - \wbf^*}^2\\& + \eta_t^2\squab{1 + \frac{K+1/\ssf{SNR}}{M}}\frac{H^2}{K}.
\end{split}
\end{equation}
\end{lemma}
\begin{proof}
We introduce an auxiliary error-free global model $\bar\wbf_{t + 1} = \frac{1}{K}\wbf_{t+1}^k$. We first have
\begin{equation}\small\label{eq.norm}\begin{split}
    & \expt\norm{\wbf_{t+1} - \wbf^*}^2 = \expt\norm{\wbf_{t+1} - \bar\wbf_{t + 1} + \bar\wbf_{t + 1}  - \wbf^*}^2\\
    & = \expt \underbrace{\norm{\wbf_{t+1} - \bar\wbf_{t + 1}}^2}_{A_1} + \expt \underbrace{\norm{\bar\wbf_{t + 1}  - \wbf^*}^2}_{A_2}\\
    & + 2\expt \underbrace{{\left<\wbf_{t+1} - \bar\wbf_{t + 1}, \bar\wbf_{t + 1}  - \wbf^*\right>}}_{A_3}.
\end{split}
\end{equation}
Note that $\expt{[A_3]} = 0$. 
Then, $\expt{[A_2]}$ can be obtained from a well-known result \cite{li2019convergence}:
\begin{equation}\small\label{eq.A2}
    \expt\norm{\wbf_{t+1} - \wbf^*}^2 \leq(1-2\mu\eta_t)\expt\norm{\wbf_t - \wbf^*}^2 + \eta_t^2\frac{H^2}{K}.
\end{equation}
We finally focus on $\expt{[A_1]}$. Based on \eqref{eq.varhkk} and (\ref{eq.varhkj}), we have
\begin{equation}\label{eq.A1}\small
\begin{split}
    & \expt\norm{\wbf_{t+1} - \bar\wbf_{t + 1}}^2 = \expt\norm{\frac{1}{K}\sum_{k\in[K]}\xbf_{k} - \frac{1}{K}\sum_{k\in[K]}\hat\xbf_{k}}^2\\
    & = \frac{1}{K^2}\expt\left\Vert\sum_{k \in [K]} \hbf_k^H \hbf_k \xbf_{k}+ \sum_{k \in [K]}\sum_{j \in [K], j\neq k} \hbf_k^H \hbf_j \xbf_{j}\right.\\
    &\left. +  \sum_{k \in [K]} \hbf_k^H\nbf_i -\sum_{k\in[K]}\xbf_{k} \right\Vert^2\\
    & = \eta_t^2\frac{\sum_{k\in [K]}\expt\norm{\nabla\tilde{f}_k(\wbf)}^2}{K^2}\left(\frac{K + 1/\ssf{SNR}}{M}\right)\\
    &
    \leq\eta_t^2\frac{H^2}{K}\left(\frac{K + 1/\ssf{SNR}}{M}\right).
\end{split}
\end{equation}
Plugging (\ref{eq.A2}) and (\ref{eq.A1}) back to (\ref{eq.norm}) completes the proof.
\end{proof}

Building on Lemma \ref{lemma.1}, we next present a complete convergence upper bound for random orthogonalization. Due to space limitation, the proof of Theorem~\ref{thm.1} is omitted and will be reported in the journal version.

\begin{theorem}[\textbf{\em Convergence for random orthogonalization}]
\label{thm.1} 
With Assumptions 1-4, for some $\gamma\geq 0$, if we select the learning rate as $\eta_t = \frac{2}{\mu(t+\gamma)}$, we have
\begin{equation}\label{eq.convergence}
    \expt[f(\wbf_t)] - f^* \leq \frac{L}{2(t + \gamma)}\left[\frac{4B}{\mu^2} + (1 + \gamma)\norm{\wbf_0 - \wbf^*}^2\right],
\end{equation}
for any $t\geq 1$, where 
\begin{equation}\label{eq.Bdef}
B\triangleq \squab{1 + \frac{K+1/\ssf{SNR}}{M}}\frac{H^2}{K}.
\end{equation}
\end{theorem}


Lemma \ref{lemma.1} and Theorem \ref{thm.1} illustrate that there are two main factors that influence the convergence rate of FL: \textbf{variance reduction} and \textbf{channel interference}. In particular, the definition of $B$ in \eqref{eq.Bdef}, which appears in both Lemma \ref{lemma.1} and Theorem \ref{thm.1}, captures the joint impact of both factors.  The nature of distributed SGD suggests that, for a fixed mini-batch size at each client, involving $K$ devices enjoys a $\frac{1}{K}$ variance reduction of stochastic gradient at each SGD iteration \cite{johnson2013accelerating}, which is captured by the $\frac{H^2}{K}$ term in \eqref{eq.thrm1} and \eqref{eq.convergence}. However, due to the existence of interference, the convergence rate is determined by both variance reduction and channel interference, shown as $\frac{H^2}{K}$ and $\frac{(K + 1/\ssf{SNR})H^2}{MK}$ terms in \eqref{eq.Bdef}. This suggests that the desired variance reduction may be adversely impacted if channel interference dominates the convergence bound. In particular, when $M>>K$, we have $\frac{1}{K}>>\frac{K+1/\ssf{SNR}}{MK}$, and the system enjoys almost the same variance reduction as the interference-free case. However, in the case of $K>>M$, we have $\frac{(K + 1/\ssf{SNR})}{MK}\approx \frac{1}{M}>>\frac{1}{K}$, and $\frac{H^2}{M}$ dominates the convergence bound. In this case, blindly increasing the number of clients is unwise, as it does not bring the advantage of variance reduction. 

\begin{remark}
In massive MIMO, a BS is usually equipped with hundreds of antennas. Although there may exist large number of users participating in FL, only a small number of them are simultaneously active \cite{yang2020federated}, especially in millimeter wave cells whose coverage are usually small. Both factors indicate that $K<<M$ often holds in typical massive MIMO systems. The  analysis reveals that our proposed framework enjoys nearly the same interference-free convergence rate with low communication and computation overhead in massive MIMO systems.
\end{remark} 

\begin{figure}
    \centering
    \includegraphics[width = 0.9\linewidth]{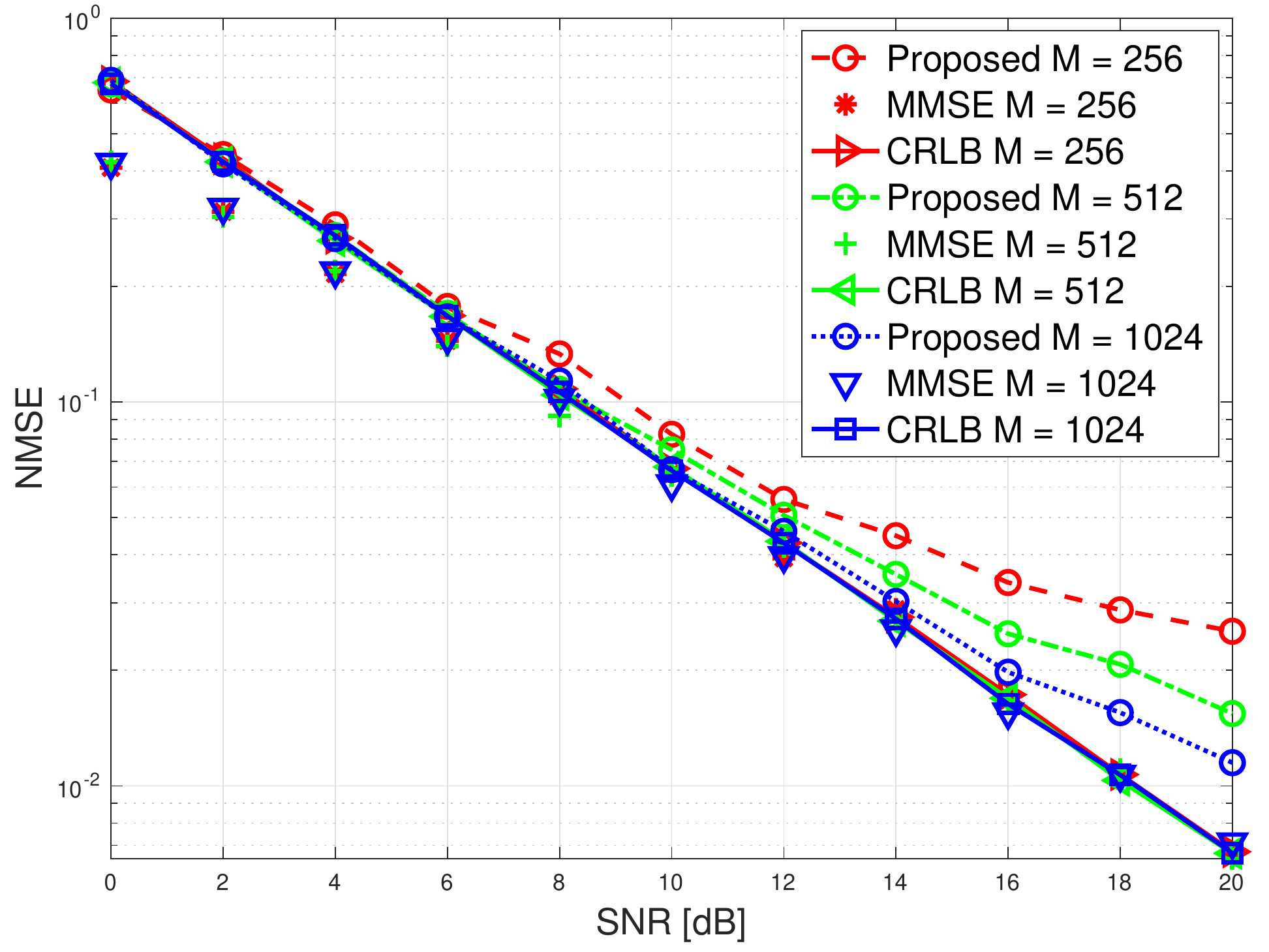}
    \caption{NMSE of the global ML model parameters versus SNR.}
    \label{fig:LP}
\end{figure}
\begin{figure}
    \centering
    \includegraphics[width = 0.9\linewidth]{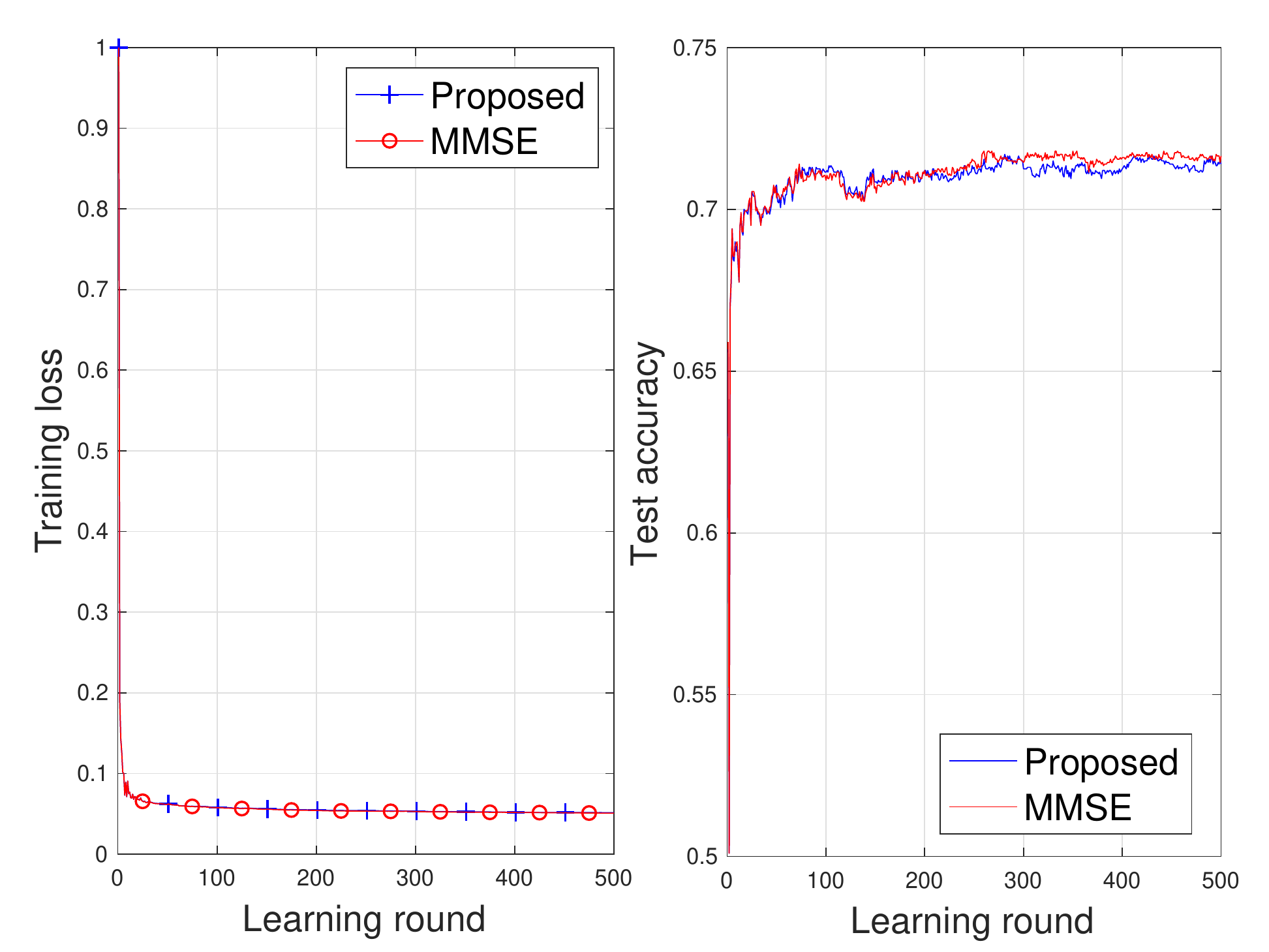}
    \caption{Training loss and test accuracy of a SVM FL task of random orthogonalization and MMSE.}
    \label{fig:SVM}
\end{figure}

\section{Experiments}
\label{sec:sim}

We evaluate the performances of random orthogonalization for uplink FL communications through numerical experiments. From a communication performance perspective, we compare the proposed method with the traditional MIMO detector to compute the global model. Then, we use a real-world FL task to evaluate the learning performance of the proposed method.

\subsection{Communication performance}
We consider a massive MIMO BS with $M = 256, 512$, and $1024$ antennas, with $K = 8$ active users participating in a FL task. We assume a Rayleigh fading channel model, i.e., $\hbf_k \sim \mathcal{CN}(0,\frac{1}{M}\Ibf)$, for each user, and use the normalized mean square error (NMSE) of the computed global model to evaluate the system performance. All NMSE results in the simulation are obtained from $2,000$ Monte Carlo experiments. Fig.~\ref{fig:LP} compares the NMSE performance of the proposed scheme with a MMSE decoder as well as CRLB under different system SNRs. As illustrated in Fig.~\ref{fig:LP}, the proposed method performs nearly identically {to} the MMSE decoder in low and moderate SNRs under different antenna configurations (see $\ssf{SNR} \leq 12$ dB). As the SNR increases, the dominant factor affecting system performances {becomes} the interference among different users. Unlike the MMSE decoder that can cancel all interferences when $K\leq M$ at high SNR, Eqn.~\eqref{eq:SINR} shows that, for a given $K$ and $M$, the proposed framework has a fixed (approximate) $\ssf{SIR} = \frac{K - 1}{M}$ as $\ssf{SNR} \rightarrow \infty$, which explains why the performance of the proposed scheme deteriorates compared with MMSE at high SNR. However, this issue disappears naturally as the number of BS antennas increases. It can be seen in Fig.~\ref{fig:LP} that the performance gap between the proposed method and MMSE reduces, from about $5$ dB when $M = 256$ to about $2$ dB when $M = 1024$ at $\ssf{SNR} = 20$ dB. Note that our method only requires $1/K$ of channel estimation overheard compared with MMSE, and this advantage is more significant when the BS is equipped with larger number of antennas.

\begin{table}
\caption{Computation time comparison between random orthogonalization and MMSE Decoder}
\label{table:comp}
    \centering
    \begin{tabular}{c| c c c}
    \hline
    \# antennas
    &   \multicolumn{3}{c}{Total CPU time (second)}\cr
    (M) & Proposed & MMSE & Proposed/MMSE\cr \hline
    \hline
    256 & 0.0186 & 2.7141 & 0.68\% \cr
    512 & 0.0303 & 12.4155 & 0.24\% \cr
    1024 & 0.0448 & 82.3530 & 0.05\% \cr\hline
    \end{tabular}
\end{table}

We next focus on the low-latency benefit of random orthogonalization. Table \ref{table:comp} compares the computation time of the proposed scheme and MMSE decoder with $\ssf{SNR} = 10$ dB. The total CPU time is the cumulative time {of each algorithm} over $2,000$ Monte Carlo experiments. We see that the time consumption of random orthogonalization is less than $1\%$ of the MMSE baseline. Especially, when $M = 1024$, despite the $0.3$ dB NMSE performance loss compared with the MMSE decoder (as shown in Fig.~\ref{fig:LP}), the computation time of the proposed method is only $0.05\%$ of the MMSE baseline. The results suggest that the random orthogonalization framework is attractive in massive MIMO systems, because it has nearly identical NMSE performance to CRLB but requires much less channel estimation overhead and achieves extremely lower system latency than the MMSE decoder.

\subsection{Learning performance}
We use a classification task to evaluate the ML model accuracy and convergence rate of the proposed random orthogonalization approach. In particular, we implement a support vector machine (SVM) to classify even and odd numbers in the MNIST handwritten-digit dataset \cite{deng2012mnist}, with $d = 784$. We consider a BS with $M = 256$ antennas and $K=8$ active clients involved in this task. The size of the local training set at each client is $500$, the size of the test set is $2,000$, and we set $E = 1$. Fig.~\ref{fig:SVM} reports the training loss and test accuracy of the proposed method and MMSE decoder with SNR $= 10 $ dB. Although the MSE of the global model at the BS during the learning process is about $2$ dB worse for random orthogonalization as shown in Fig.~\ref{fig:LP}, the actual learning performances of the two methods are nearly identical, further validating the effectiveness of random orthogonalization.

\section{Conclusion}
\label{sec:conc}

Leveraging the unique characteristics of channel hardening and favorable propagation in massive MIMO, we have proposed a novel uplink communication and processing method, coined \emph{random orthogonalization}, that significantly reduces the channel estimation overhead while achieving natural over-the-air model aggregation without requiring transmitter side channel state information. Theoretical performance analyses, from both communication (CRLB) and machine learning (model convergence rate) perspectives, have been carried out. The theoretical results suggested that random orthogonalization asymptotically achieves the same convergence rate as vanilla FL with perfect communications, and were further validated with numerical experiments. More importantly, random orthogonalization improves the scalability of FL, which is a critical feature that is often bottlenecked by the limited wireless resources.


\bibliographystyle{IEEEtran}
\bibliography{wireless,Shen,FedLearn,ref}

\end{document}